\documentclass[12pt]{article}

\usepackage[a4paper]{geometry}
\usepackage{amssymb}
\usepackage{amsmath}
\usepackage{amsthm}

\theoremstyle{plain}
\newtheorem{theorem}{Theorem}[section]
\newtheorem{proposition}[theorem]{Proposition}
\newtheorem{lemma}[theorem]{Lemma}
\newtheorem{corollary}[theorem]{Corollary}

\theoremstyle{definition}

\newtheorem*{memo*}{MEMO}

\theoremstyle{remark}

\newtheorem*{remark*}{Remark}

\newtheorem*{remarks*}{Remarks}

\numberwithin{equation}{section}

\title{Hypergeometric expression for the resolvent of the discrete Laplacian in low dimensions}
\author{Kenichi {\scshape Ito}\footnote{Graduate School of Mathematical Sciences, 
The University of Tokyo, 3-8-1 Komaba, Meguro-ku, Tokyo 153-8914, Japan.
E-mail: \texttt{ito@ms.u-tokyo.ac.jp}. 
}
\ \& 
Arne {\scshape Jensen}\footnote{Department of Mathematical Sciences,
Aalborg University, Skjernvej 4A, DK-9220 Aalborg \O{}, Denmark.
E-mail: \texttt{matarne@math.aau.dk}. 
}}
\date{}

\begin{document}
\allowdisplaybreaks
\maketitle

\begin{abstract}
We present an explicit formula for the resolvent of the discrete Laplacian on the square lattice,
and compute its asymptotic expansions around thresholds in low dimensions. 
As a by-product we obtain a closed formula 
for the fundamental solution to the discrete Laplacian. 
For the proofs we express the resolvent in a general dimension in terms of the 
Appell--Lauricella hypergeometric function of type $C$ outside a disk encircling the spectrum.
In low dimensions it reduces to a generalized hypergeometric function, 
for which certain transformation formulas are available for the desired expansions. 
\end{abstract} 


\section{Introduction}\label{191229}

Let $d\in\mathbb N=\{1,2,\ldots\}$,
and define $G\colon (\mathbb C\setminus [0,4d])\times \mathbb Z^d\to\mathbb C$ as 
\begin{align}
\begin{split}
G(z,n)
&=
(2\pi)^{-d}\int_{\mathbb T^d}\frac{\mathrm e^{\mathrm in\theta}}{2d-2\cos\theta_1-\dots-2\cos\theta_d-z}\,\mathrm d\theta
.
\end{split}
\label{11.4.17.3.5b}
\end{align}
We are interested in the asymptotic behavior of $G(z,n)$ 
as $z$ approaches one of the \emph{thresholds} $0,4,8,\cdots,4d\in [0,4d]$.
In fact, the main results of the paper present the asymptotic expansions for dimensions $d=1,2$. 
In our previous paper \cite{IJ} we explicitly determined 
all the singular parts of expansions for all the dimensions
in terms of the \emph{Appell--Lauricella hypergeometric function of type $B$}, $F^{(d)}_B$: 
For 
each threshold $4q$, $q=0,\dots,d$, there exist functions 
$\mathcal E_q(z,n)$ and $\mathcal F_q(z,n)$
analytic in $|z-4q|<4$ such that 
\begin{align}
G(z,n)=\mathcal E_q(z,n)+f_q(z)\mathcal F_q(z,n)
\label{191226}
\end{align}
where
$$f_q(z)=\begin{cases}
(z-4q)^{(d-2)/2}&\text{if $d$ is odd},\\
(z-4q)^{(d-2)/2}\log(z-4q)&\text{if $d$ is even}.
\end{cases}$$ 
Moreover, $\mathcal F_q(z,n)$ is explicitly written in terms of $F^{(d)}_B$,
see Appendix~\ref{190831}.
However, thus far there seem to be no explicit results on $\mathcal E_q(z,n)$ in the literature.

The above problem is originally motivated by  the resolvent kernel of 
the discrete Laplacian. 
Define $H_0\colon\ell^2(\mathbb Z^d)\to \ell^2(\mathbb Z^d)$ as, 
for any $u\in \ell^2(\mathbb Z^d)$, 
\begin{equation*}
(H_0u)[n]=\sum_{j=1}^d\bigl(2u[n]-u[n+e_j]-u[n-e_j]\bigr)
\ \ \text{for } n\in\mathbb Z^d,
\end{equation*}
where $\{e_j\}_{j=1,\dots,d}\subset \mathbb Z^d$ is the standard basis.
Then the operator $H_0$ is bounded and self-adjoint on the Hilbert space $\ell^2(\mathbb Z^d)$,
and its spectrum is given by $\sigma(H_0)=[0,4d]$.
It is well known, see e.g.\ \cite{IJ}, that its resolvent
$R_0(z)=(H_0-z)^{-1}$ is a convolution operator by $G(z,\cdot)$: 
For any $z\in \mathbb C\setminus [0,4d]$ and $u\in \ell^2(\mathbb Z^d)$ 
$$R_0(z)u=G(z,\cdot)*u.$$

The asymptotic expansion of a \emph{perturbed} resolvent around a threshold
determines the time-decay rate of the Schr\"odinger propagator,
and, in connection with \emph{threshold resonances}, has been widely studied 
in the Schr\"odinger theory on the Euclidean space. 
The  results basically rely on the explicit expansion of the \emph{free} resolvent,
see \cite{JN1,JN2} for the standard strategy. 
However, in multi-dimensional discrete spaces, almost nothing explicit is known for the free resolvent,
and hence the analysis of thresholds so far has to somehow avoid it,
see, e.g., \cite{IK,NT} for recent results. 
We remark that, on the other hand, 
a complete analysis is possible for the one-dimensional discrete case \cite{IJ1}.

As another motivation, let us introduce a probabilistic construction 
of a fundamental solution to 
$H_0$, though it is quite elementary, see also \cite{LL,MW}.
Let $X_k\in\mathbb Z^d$ be the position of a random walker on $\mathbb Z^d$ 
at time $k\in\mathbb N_0:=\{0\}\cup\mathbb N$, starting from $0\in\mathbb Z^d$ at time $k=0$.
Then the expectation of the number of times that the walker visits $n\in \mathbb Z^d$ is given by 
\begin{align*}
\mathbb E[n]=\sum_{k=0}^\infty P(X_k=n),
\end{align*}
and this \emph{would} provide a fundamental solution to $H_0$ up to a constant factor:
\begin{align*}
H_0\mathbb E=2d \delta_0.
\end{align*} 
In fact, $2d\mathbb E[n]$, $n\neq 0$, should coincide with sum of expectations of random walkers
starting from $\pm e_1,\dots,\pm e_d\in\mathbb Z^d$,
while  
$2d\mathbb E[0]$ should exceed the sum by $2d$ because the walker at 
$n=0$ at time $k=0$ is excessively counted. 
However, $\mathbb E[n]$ is known to diverge for $d=1,2$, 
cf.\ \emph{recurrence and transience} of random walk, and we need a certain \emph{renormalization}. 
For that we let the walker \emph{disappear} in each unit time with probability $\epsilon\in(0,1]$, 
so that 
\begin{align*}
\mathbb E(\epsilon,n)=
\sum_{k=0}^\infty \left(1-\epsilon\right)^kP(X_k=n).
\end{align*}
This is obviously convergent for all $d\in\mathbb N$, and 
furthermore we can show that 
\begin{align}
\mathbb E(\epsilon,n)
=
\frac{2d}{1-\epsilon} G\left(\frac{-2d\epsilon}{1-\epsilon},n\right)
,
\label{190805}
\end{align}
see Appendix~\ref{191219}. 
Hence the analysis of $\mathbb E(\epsilon,n)$ as $\epsilon\to 0$ reduces to that of $G(z,n)$
as $z\to 0$. 
Note that then, according to \cite{IJ} or Appendix~\ref{190831}, 
the following renormalization should work: 
$$
\lim_{\epsilon \to +0}\left[\mathbb E(\epsilon,n)-e(\epsilon)\right]
=2d \mathcal E_0(0,n);\quad 
e(\epsilon)=
\begin{cases}
(2\epsilon )^{-1/2}& \text{if }n=1,\\
-(1/\pi)\log (4\epsilon)& \text{if }n=2,\\
0& \text{otherwise,}
\end{cases}
$$
where $\mathcal E_0$ is from \eqref{191226},
being compatible with recurrence and transience properties.

In fact, $\mathcal E_0(0,n)$ is a fundamental solution to $H_0$,
and we can deduce its expression for $d=2$ immediately from the expansion.
A fundamental solution is a basic tool in analysis 
of discretized differential equations \cite{BO,K,RS}.
There is already a practical procedure to compute its value at each $n\in\mathbb Z^2$, see e.g.\ \cite{MW,Mo},
and moreover its precise far field asymptotics can be derived directly 
from the integral expression \cite{L,Man,MW}.
However, a closed formula seems to be still missing, and it itself would be of theoretical interest.

This paper is organized as follows. 
In Section~\ref{1912261716} we present the main results of the paper,
while their proofs are given in Section~\ref{18111121b}. 
In Appendix~\ref{19121823}
we compare the singular parts of the expansions in this paper with those of the previous work \cite{IJ},
since their coincidence is not clear at a glance. 
Some non-trivial identities are obtained there. 
Appendix~\ref{191219} contains a short discussion proving \eqref{190805}.

\section{Main results}\label{1912261716}

\subsection{Generalizations of hypergeometric functions}\label{191218}

Here we recall generalizations of hypergeometric functions.
The \emph{ordinary hypergeometric function},
or simply \emph{hypergeometric function}, $F$ is defined as an analytic continuation of 
\begin{align*}
F\left(\genfrac{}{}{0pt}{}{a,b}{c};w\right)
=\sum_{k=0}^\infty\frac{(a)_k(b)_k}{(c)_k}\frac{w^k}{k!}, 
\end{align*}
where $(q)_j$ denotes the \emph{Pochhammer symbol}: 
\begin{align*}
(q)_j=\left\{\begin{array}{ll}1&\text{for }j=0,\\q(q+1)\cdots(q+j-1)&\text{for }j=1,2,\dots.\end{array}\right.
\end{align*}
One generalization we shall use is  
the \emph{generalized hypergeometric function} ${_pF_q}$ defined as an analytic continuation of 
\begin{align*}
{_pF_q}\left(\genfrac{}{}{0pt}{}{a_1,\dots,a_p}{b_1,\dots,b_q};w\right)
=\sum_{k=0}^\infty\frac{(a_1)_k\cdots(a_p)_k}{(b_1)_k\cdots(b_q)_k}\frac{w^k}{k!}.
\end{align*}
Other ones are  
the \emph{$d$-dimensional Appell--Lauricella hypergeometric functions of type $B$, $C$},
which are defined as analytic continuations of 
\begin{align*}
&
F^{(d)}_B\left(\genfrac{}{}{0pt}{}{a_1,\dots,a_d;b_1,\dots,b_d}c;w_1,\dots,w_d\right)
\\&
=\sum_{\alpha\in\mathbb N_0^d}\frac{(a_1)_{\alpha_1}\cdots(a_d)_{\alpha_d}(b_1)_{\alpha_1}\cdots(b_d)_{\alpha_d}}{(c)_{\alpha_1+\dots+\alpha_d}}
\frac{w_1^{\alpha_1}\cdots w_d^{\alpha_d}}{\alpha_1!\cdots \alpha_d!}
,
\\
&F_C^{(d)}\left(\genfrac{}{}{0pt}{}{a,b}{c_1,\dots,c_d};w_1,\dots,w_d\right)
=
\sum_{\alpha\in\mathbb N_0^d}
\frac{(a)_{\alpha_1+\dots+\alpha_d}(b)_{\alpha_1+\dots+\alpha_d}}{(c_1)_{\alpha_1}\cdots(c_d)_{\alpha_d}}
\frac{w_1^{\alpha_1}\cdots w_d^{\alpha_d}}{\alpha_1!\cdots \alpha_d!}
,
\end{align*}
respectively. 
Let us remark that it is straightforward to see
\begin{align}
F\left(\genfrac{}{}{0pt}{}{a,b}c;w\right)
={_2F_1}\left(\genfrac{}{}{0pt}{}{a,b}c;w\right)
=F^{(1)}_B\left(\genfrac{}{}{0pt}{}{a,b}c;w\right)
=F^{(1)}_C\left(\genfrac{}{}{0pt}{}{a,b}c;w\right)
.
\label{191221}
\end{align}

\subsection{Laurent expansion in general dimension}\label{19122616}

We first present the Laurent expansion of $G(z,n)$ 
outside a disk encircling the interval $[0,4d]\subset \mathbb C$. 

\begin{theorem}\label{181110}
For any $z\in\mathbb C$ with $|2d-z|>2d$ one has an expression
\begin{align*}
G(z,n)
&
=
\sum_{\alpha\in \mathbb N_0^d}
\frac{ (2|\alpha|+|n|)!}{\alpha!\prod_{j=1}^d[(\alpha_j+|n_j|)!]}
(2d-z)^{-2|\alpha|-|n|-1}
,\intertext{or equivalently,}
G(z,n)
&=
\frac{|n|!}{|n_1|!\cdots|n_d|!}
\frac1{(2d-z)^{|n|+1}}
\\&\phantom{{}={}}{}\cdot
F_C^{(d)}\biggl(\genfrac{}{}{0pt}{}{\frac{|n|+1}2,\frac{|n|+2}2}{
|n_1|+1,\dots,|n_d|+1};
\frac4{(z-2d)^2},\dots,\frac4{(z-2d)^2}\biggr)
.
\end{align*}
where $|n|=|n_1|+\dots+|n_d|$.
\end{theorem}

Despite the explicit Laurent expansion 
it is highly non-trivial to re-expand it around each threshold for $d\ge 2$.
Transformation formulas for $F_C^{(d)}$ are not widely studied, and 
there seem to be no transformation formulas appropriate for our purpose, as far as we are aware,
cf.\ \cite{NIST:DLMF,E,S}.
Fortunately in low dimensions, we can reduce $F_C^{(d)}$
to some of ${_pF_q}$, for which many more transformation formulas are available.
The reduction for $d=1$ is clear from \eqref{191221},
and the one for $d=2$ will be given in Corollary~\ref{1908056}.

\subsection{Endpoint thresholds in dimension one}\label{19122617}

Let us start with the $1$-dimensional case. 
We already essentially obtained the expansion of the $1$-dimensional resolvent 
in \cite[Proposition~2.1]{IJ1}, 
and we do not consider the following formulas as really new. 
However, we have not found the same closed expression in the literature,
hence it would be worth while to exhibit them here.

\begin{theorem}\label{190716}
Let $d=1$. Then for any $z\in\mathbb C\setminus[0,4]$ and $n\in\mathbb Z$
\begin{align}
\begin{split}
G(z,n)
&
=\frac{\left(-z+2-\sqrt{z(4-z)}\right)^{|n|}}{2^{|n|}\sqrt{z(4-z)}}
.
\end{split}
\label{190807}
\end{align}
In particular, for any $z\in\mathbb C\setminus[0,\infty)$ and $n\in\mathbb Z$
\begin{align}
G(z,n)
&
=
-\frac{|n|}{2}
F\left(\genfrac{}{}{0pt}{}{1+n,1-n}{\frac32};\frac{z}4\right)
+\frac1{2\sqrt{-z}}
F\left(\genfrac{}{}{0pt}{}{\frac12+n,\frac12-n}{\frac12};\frac{z}4\right),
\label{19083117}
\intertext{and for any $z\in\mathbb C\setminus(-\infty,4]$ and $n\in\mathbb Z$}
\begin{split}
G(z,n)&
=
-\frac{(-1)^{n+1}|n|}{2}
F\left(\genfrac{}{}{0pt}{}{1+n,1-n}{\frac32};\frac{4-z}4\right)
\\&\phantom{{}={}}{}
+\frac{(-1)^{n+1}}{2\sqrt{z-4}}
F\left(\genfrac{}{}{0pt}{}{\frac12+n,\frac12-n}{\frac12};\frac{4-z}4\right)
.
\end{split}
\label{19083118}
\end{align}
\end{theorem}
\begin{remarks*}
\begin{enumerate}
\item
Here and below the branch of $\sqrt w$ is the principal one
with the cut along the negative real axis,
i.e., $\mathop{\mathrm{Re}}\sqrt w>0$ for $w\in \mathbb C\setminus (-\infty,0]$.
\item
It is clear that the zeroth order terms of the analytic parts of \eqref{19083117} and \eqref{19083118}
are fundamental solutions to $H_0$ and $H_0-4$, respectively. 
\end{enumerate}
\end{remarks*}

\subsection{Embedded threshold in dimension two}\label{19122618}

On the $2$-dimensional lattice $\mathbb Z^2$ we have different types of 
expansions for the embedded threshold $z=4$ and the endpoint ones $z=0,8$.  
The former is completely and clearly determined as follows.

\begin{theorem}\label{19083122}
Let $d=2$. For any $|z-4|<4$ with $\mathop{\mathrm{Im}}z>0$ and $n\in\mathbb Z^2$
with $|n|$ even
\begin{align*}
G(z,n)
&=
(-1)^{\max\{|n_1|,|n_2|\}}\frac{|n_1+n_2|}2\frac{|n_1-n_2|}2\frac{z-4}{4}
\\&\phantom{{}={}}{}\cdot
{_4F_3}\left(
\genfrac{}{}{0pt}{}{\frac{2+n_1+n_2}2,\frac{2-n_1-n_2}2,\frac{2+n_1-n_2}2,\frac{2-n_1+n_2}2
}{1,\frac32,\frac32}; 
\frac{(z-4)^2}{16}\right)
\\&\phantom{{}={}}{}
+\frac{\mathrm i(-1)^{\max\{|n_1|,|n_2|\}}}{4\pi}
\sum_{k=0}^\infty
\frac{
\left(\frac{1+n_1+n_2}2\right)_k
\left(\frac{1+n_1-n_2}2\right)_k
\left(\frac{1-n_1+n_2}2\right)_k
\left(\frac{1-n_1-n_2}2\right)_k
}{
k!k!\left(\frac12\right)_k\left(\frac12\right)_k
}
\\&\phantom{{}={}+{}}{}
\cdot 
\left(\frac{z-4}4\right)^{2k}
\biggl[
2\psi(1+k)
+2\psi\left(\frac12+k\right)
-\psi\left(\frac{1+n_1+n_2}2+k\right)
\\&\phantom{{}={}+\left(\frac{z-4}4\right)^{2k}\cdot \biggl[}{}
-\psi\left(\frac{1+n_1-n_2}2+k\right)
-\psi\left(\frac{1-n_1+n_2}2+k\right)
\\&\phantom{{}={}+\left(\frac{z-4}4\right)^{2k}\cdot \biggl[}{}
-\psi\left(\frac{1-n_1-n_2}2+k\right)
-\log\left(-\frac{(z-4)^2}{16}\right)
\biggr],
\intertext{and 
for any $|z-4|<4$ with $\mathop{\mathrm{Im}}z>0$ and $n\in\mathbb Z^2$ with $|n|$ odd} 
G(z,n)
&=
\frac{(-1)^{\max\{|n_1|,|n_2|\}}}4{_4F_3}\left(\genfrac{}{}{0pt}{}{
\frac{1+n_1+n_2}{2},\frac{1+n_1-n_2}{2},\frac{1-n_1+n_2}{2},\frac{1-n_1-n_2}{2}}{1,\frac12,\frac12}; 
\frac{(z-4)^2}{16}\right)
\\&\phantom{{}={}}{}
+\frac{\mathrm i(-1)^{\max\{|n_1|,|n_2|\}}}{\pi}\frac{|n_1+n_2|}2\frac{|n_1-n_2|}2\frac{z-4}4
\\&\phantom{{}={}+{}}{}\cdot
\sum_{k=0}^\infty
\frac{
\left(\frac{2+n_1+n_2}2\right)_{k}
\left(\frac{2+n_1-n_2}2\right)_{k}
\left(\frac{2-n_1+n_2}2\right)_k
\left(\frac{2-n_1-n_2}2\right)_k
}{
k!k!\left(\frac32\right)_k\left(\frac32\right)_k
}
\\&\phantom{{}={}+{}}{}\cdot
\left(\frac{z-4}4\right)^{2k}
\biggl[
2\psi\left(1+k\right)
+2\psi\left(\frac32+k\right)
-\psi\left(\frac{2+n_1+n_2}2+k\right)
\\&\phantom{{}={}+{}}\phantom{{}\cdot\left(\frac{z-4}4\right)^{2k+1}\biggl[} {}
-\psi\left(\frac{2+n_1-n_2}2+k\right)
-\psi\left(\frac{2-n_1+n_2}2+k\right)
\\&\phantom{{}={}+{}}\phantom{{}\cdot\left(\frac{z-4}4\right)^{2k+1}\biggl[} {}
-\psi\left(\frac{2-n_1-n_2}2+k\right)
-\log\left(-\frac{(z-4)^2}{16}\right)
\biggr]
.
\end{align*}
\end{theorem}
\begin{remarks*}
\begin{enumerate}
\item
Here and below the branch of $\log w$ is the principal one 
with the cut along the negative real axis,
i.e., $-\pi<\mathop{\mathrm{Im}}\log w<\pi$ for $w\in\mathbb C\setminus (-\infty,0]$. 

\item
The \emph{psi function}, or \emph{digamma function}, 
$\psi(w)$, is defined as 
\begin{align*}
\psi(w)=\frac{\Gamma'(w)}{\Gamma(w)}\ \ \text{for }w\in\mathbb C\setminus \{0,-1,-2,\ldots\},
\end{align*}
where $\Gamma$ is the gamma function, see \cite[(5.2.2)]{NIST:DLMF}. 
According to \cite[(5.4.14)]{NIST:DLMF} and \cite[(5.4.15)]{NIST:DLMF},
we have for $m=1,2,\ldots$
\begin{align*}
\psi(1+m)=-\gamma+\sum_{k=1}^m\frac1k,\quad 
\psi\left(\frac12+m\right)=-\gamma-2\log 2+2\sum_{k=1}^m\frac1{2k-1},
\end{align*}
where $\gamma$ is the Euler--Mascheroni constant.

\item
The formulas from Theorem~\ref{19083122} generalize,
e.g.,  
\cite[(36)]{Ma} to the off-diagonal vertices. See also \cite[\S 14.3 and \S 15.8]{NIST:DLMF}.

\item
The value of $G(z,n)$ at $n=0$ coincides with a branch of the 
\emph{complete elliptic integral of the first kind},
$K(w)$, up to a constant factor. In fact, for $0<|w'|<1$ with $w^2+w'^2=1$ 
\begin{align*}
K(w)=\sum_{k=0}^\infty\frac{\left(\frac12\right)_k\left(\frac12\right)_k}{k!k!}w'^{2k}
\left[\psi(1+k)-\psi\left(\frac12+k\right)-\log w'\right],
\end{align*} 
see \cite[(19.12.1)]{NIST:DLMF}. 
See also \cite{Mo}.
Hence we cannot expect 
$G(z,n)$ to be an elementary function for $d=2$ unlike \eqref{190807} for $d=1$.
\item
The corresponding fundamental solution to $H_0-4$ is given by 
\begin{align*}
\mathcal E_1(4,n)
&=
\frac{(-1)^{\max\{|n_1|,|n_2|\}}-(-1)^{\min\{|n_1|,|n_2|\}}}8
\\&\phantom{{}={}}{}
+\frac{\mathrm i[(-1)^{n_1}+(-1)^{n_2}]}{2\pi}
\left[
\log 2
-\sum_{k=1}^{|n_1+n_2|/2}\frac1{2k-1}
-\sum_{k=1}^{|n_1-n_2|/2}\frac1{2k-1}
\right]
.
\end{align*}
A renormalization as mentioned in Section~\ref{191229} is not unique, 
and for example $\log 2$ may be absorbed into the singular part. 
Moreover, it is clear that $\mathop{\mathrm{Re}}\mathcal E_1(4,\cdot)$ alone  
provides a fundamental solution,
cf.\ \cite{BO},
while $\mathop{\mathrm{Im}}\mathcal E_1(4,\cdot)$ in fact is a \emph{generalized} eigenfunction
in the sense that $\mathop{\mathrm{Im}}\mathcal E_1(4,\cdot)\notin\ell^2(\mathbb Z^2)$. 

\item 
The operator $H_0-4$ is unitarily equivalent to the \emph{discrete d'Alembertian} 
$$(\square u)[n]=u[n+e_1]+u[n-e_1]-u[n+e_2]-u[n-e_2].$$
In fact, we have 
\begin{align}
U^*(H_0-4)U=\square;\quad 
(Uu)[n]=(-1)^{n_1}u[n].
\label{19122915}
\end{align}
Hence in particular we have a fundamental solution to $\square$.
\end{enumerate}
\end{remarks*}

\subsection{Endpoint thresholds in dimension two}\label{19122619}

Expansions around the endpoint thresholds $z=0,8$ are more involved.
We change variables as, for $(m,l)\in\mathbb Z^2$,  
\begin{align*}
\mathcal P(m,l)&=(-1)^{m+l}G(z,m+l,m-l),
\\
\mathcal Q(m,l)&=(-1)^{m+l}G(z,m+l+1,m-l),
\end{align*}
 suppressing the dependence on $z\in\mathbb C\setminus [0,8]$.
We further set the diagonal values  
\begin{align}
\mathcal P_0(m)=\mathcal P(m,0)=(-1)^{m}G(z,m,m).
\label{19122923}
\end{align}
By symmetry it suffices to state the result only for $(m,l)\in\mathbb N_0^2$, see \eqref{1909039}.

\begin{theorem}\label{19080717}
Let $d=2$. 
For any $z\in\mathbb C\setminus[0,8]$ with $|z(8-z)|<16$ 
and $\pm\mathop{\mathrm{Re}}(4-z)>0$
and $m\in\mathbb Z$
\begin{align*}
\mathcal P_0(m)
&
=
\pm\frac{(-1)^m}{4\pi}
\sum_{k=0}^\infty
\frac{\left(\frac12+m\right)_k\left(\frac12-m\right)_k}{k!k!}
\left(\frac{z(8-z)}{16}\right)^k
\\&\phantom{{}={}}{}
\cdot
\left[
2\psi(1+k)
-\psi\left(\frac12+m+k\right)
-\psi\left(\frac12-m+k\right)
-\log\frac{z(z-8)}{16}
\right]
,
\end{align*}
respectively.
Furthermore, for any $z\in\mathbb C\setminus [0,8]$ and $(m,l)\in\mathbb N_0^2$ 
\begin{align*}
\mathcal P(m,l)
&
=
\frac{z-4}{4}ml
{_4F_3}\left(\genfrac{}{}{0pt}{}{1+m,1-m,1+l,1-l}{1,\frac32,\frac32}; \frac{(z-4)^2}{16}\right)
\\&\phantom{{}={}}{}
+\mathcal P_0(0){_4F_3}\left(\genfrac{}{}{0pt}{}{m,-m,l,-l}{1,\frac12,\frac12}; \frac{(z-4)^2}{16}\right)
\\&\phantom{{}={}}{}
+\sum_{\mu=1}^{m}\bigl[\mathcal P_0(\mu)-\mathcal P_0(\mu-1)\bigr]
{_4F_3}\left(\genfrac{}{}{0pt}{}{1+m-\mu,\mu-m,l,-l}{1,\frac12,\frac12}; \frac{(z-4)^2}{16}\right)
\\&\phantom{{}={}}{}
+\sum_{\nu=1}^{l}\bigl[\mathcal P_0(\nu)-\mathcal P_0(\nu-1)\bigr]
{_4F_3}\left(\genfrac{}{}{0pt}{}{m,-m,1+l-\nu,\nu-l}{1,\frac12,\frac12}; \frac{(z-4)^2}{16}\right)
,
\intertext{and }
\mathcal Q(m,l)
&
=
-\frac14{_4F_3}\left(\genfrac{}{}{0pt}{}{1+m,-m,1+l,-l}{1,\frac12,\frac12}; \frac{(z-4)^2}{16}\right)
\\&\phantom{{}={}}{}
+\frac{z-4}4(2m+1)(2l+1)\mathcal P_0(0)
{_4F_3}\left(\genfrac{}{}{0pt}{}{1+m,-m,1+l,-l}{1,\frac32,\frac32}; \frac{(z-4)^2}{16}\right)
\\&\phantom{{}={}}{}
-
\frac{z-4}4(2l+1)
\sum_{\mu=-m}^{m}
\mathcal P_0(\mu)
\\&\phantom{{}={}\quad}{}
\cdot 
{_4F_3}\left(\genfrac{}{}{0pt}{}{1+m-|\mu|,|\mu|-m,1+l,-l}{1,\frac12,\frac32}; \frac{(z-4)^2}{16}\right)
\\&\phantom{{}={}}{}
-
\frac{z-4}4(2m+1)
\sum_{\nu=-l}^{l}
\mathcal P_0(\nu)
\\&\phantom{{}={}\quad}{}
\cdot 
{_4F_3}\left(\genfrac{}{}{0pt}{}{1+m,-m,1+l-|\nu|,|\nu|-l}{1,\frac12,\frac32}; \frac{(z-4)^2}{16}\right)
.
\end{align*}
\end{theorem}

\begin{remarks*}
\begin{enumerate}
\item
The expression for $\mathcal P_0(m)$ above is considered a simultaneous expansion around both $z=0,8$.
On the other hand, each of the above ${_4F_3}$ is a finite sum, 
since one of its upper parameters is always a non-positive integer. 
Hence the above ${_4F_3}$ in principle can be 
re-expanded around $z=0,8$,
which in turn provides expansions of $\mathcal P(m,l)$ and $\mathcal Q(m,l)$. 
Currently we do not know how we could further simplify the above expressions.

\item
The value $\mathcal P_0(0)$ coincides up to a constant factor with a branch of 
the complete elliptic integral of the first kind, see the remarks to Theorem~\ref{19083122}. 

\item
For a fundamental solution
it is rather standard to first compute the diagonal values and then proceed to 
the off-diagonal by symmetry and a recurrence relation,
see \cite{MW,Mo,BO,L} and references therein. 
Our strategy for the resolvent kernel is the same, however, 
we will solve the recurrence relation till the end. 
\end{enumerate}
\end{remarks*}

Let us exhibit a fundamental solution to $H_0$ as a corollary below.

\begin{corollary}\label{200103}
Define $E\colon \mathbb Z^2\to \mathbb R$ as,
for $n\in\mathbb Z^2$ with $|n|$ even 
\begin{align*}
E[n]
&
=
-(-1)^{n_1}\frac{|n_1+n_2|}2\frac{|n_1-n_2|}2
{_4F_3}\left(\genfrac{}{}{0pt}{}{\frac{2+n_1+n_2}2,\frac{2+n_1-n_2}2,\frac{2-n_1+n_2}2,\frac{2-n_1-n_2}2}{1,\frac32,\frac32}; 1\right)
\\&\phantom{{}={}}{}
-\frac{(-1)^{n_1}}{\pi}
\sum_{\mu=1}^{|n_1+n_2|/2}
(-1)^{\mu}
\left[
\sum_{j=1}^{\mu-1}\frac2{2j-1}
+\frac1{2\mu-1}
\right]
\\&\phantom{{}={}\quad}{}\cdot
{_4F_3}\left(\genfrac{}{}{0pt}{}{\frac{|n_1+n_2|}2+1-\mu,\mu-\frac{|n_1+n_2|}2,\frac{n_1-n_2}2,\frac{-n_1+n_2}2}{1,\frac12,\frac12}; 1\right)
\\&\phantom{{}={}}{}
-\frac{(-1)^{n_1}}{\pi}
\sum_{\nu=1}^{|n_1-n_2|/2}
(-1)^{\nu}
\left[
\sum_{j=1}^{\nu-1}\frac2{2j-1}
+\frac1{2\nu-1}
\right]
\\&\phantom{{}={}\quad}{}\cdot
{_4F_3}\left(\genfrac{}{}{0pt}{}{\frac{n_1+n_2}2,\frac{-n_1-n_2}2,\frac{|n_1-n_2|}2+1-\nu,\nu-\frac{|n_1-n_2|}2}{1,\frac12,\frac12}; 1\right)
,
\intertext{and for $n\in\mathbb Z^2$ with $|n|$ odd}
E[n]
&
=
\frac{(-1)^{\max\{|n_1|,|n_2|\}}}4{_4F_3}\left(\genfrac{}{}{0pt}{}{\frac{1+n_1+n_2}2,\frac{1+n_1-n_2}2,\frac{1-n_1+n_2}2,\frac{1-n_1-n_2}2}{1,\frac12,\frac12}; 1\right)
\\&\phantom{{}={}}{}
+\frac{(-1)^{\max\{|n_1|,|n_2|\}}|n_1-n_2|}{\pi}
\sum_{\mu=-(|n_1+n_2|-1)/2}^{(|n_1+n_2|-1)/2}
(-1)^\mu
\sum_{j=1}^{|\mu|}\frac1{2j-1}
\\&\phantom{{}={}\quad}{}
\cdot 
{_4F_3}\left(\genfrac{}{}{0pt}{}{\frac{1+|n_1+n_2|}2-|\mu|,|\mu|+\frac{1-|n_1+n_2|}2,\frac{1+n_1-n_2}2,\frac{1-n_1+n_2}2}{1,\frac12,\frac32}; 1\right)
\\&\phantom{{}={}}{}
+\frac{(-1)^{\max\{|n_1|,|n_2|\}}|n_1+n_2|}{\pi}
\sum_{\nu=-(|n_1-n_2|-1)/2}^{(|n_1-n_2|-1)/2}
(-1)^\nu
\sum_{j=1}^{|\nu|}\frac1{2j-1}
\\&\phantom{{}={}\quad}{}
\cdot 
{_4F_3}\left(\genfrac{}{}{0pt}{}{\frac{1+n_1+n_2}2,\frac{1-n_1-n_2}2,\frac{1+|n_1-n_2|}2-|\nu|,|\nu|+\frac{1-|n_1-n_2|}2}{1,\frac12,\frac32}; 1\right)
.
\end{align*}
Then $E$ is a fundamental solution to $H_0$,
i.e., $H_0E=\delta_0$.
\end{corollary}
\begin{remarks*}
\begin{enumerate}
\item
To construct a fundamental solution to $H_0-8$ 
we can make use of the unitary transformation \eqref{19122915} 
not only for $n_1$ but also for $n_2$-variable. We omit the details. 

\item
 Generally a fundamental solution is not unique, but, up to a sign, 
the above $E[n]$ coincides with those considered so far in literature,
since they satisfy the same difference equation with the same diagonal values. 
See e.g.\ \cite[(5)]{L} and \cite[(13)]{Man}.
The discrete fundamental solution $E[n]$ is known to approach 
a continuous one as $|n|\to\infty$ \cite{L,Man,MW}.

\item
From the above expression we can deduce another fundamental solution of a simpler form. 
In fact, we can remove the irrational parts with factor $1/\pi$, 
which forms a generalized eigenfunction, since $H_0$ acts separately on the rational and irrational parts. 
However, the resulting fundamental solution does not have a clean asymptotics as $|n|\to \infty$ 
unlike the above $E[n]$ does.
\end{enumerate}
\end{remarks*}

\section{Proofs}\label{18111121b}

In this section we present proofs of the main results of the paper.

\subsection{Proof of Laurent expansion}\label{19122621}

\begin{proof}[Proof of Theorem~\ref{181110}]
We prove only the first expression, 
since the second one follows immediately from the first and the definition of $F_C^{(d)}$. 
Using the integral expression \eqref{11.4.17.3.5b}, we compute 
\begin{align*}
G(z,n)
&=
(2\pi)^{-d}(2d-z)^{-1}\int_{\mathbb T^d}
\mathrm e^{\mathrm in\theta}
\biggl(1-\frac{2\cos\theta_1+\cdots+2\cos\theta_d}{2d-z}\biggr)^{-1}
\,\mathrm d\theta
\\&
=
(2\pi)^{-d}
\sum_{\alpha\in\mathbb N_0^d}
\frac{|\alpha|!}{\alpha!}
(2d-z)^{-|\alpha|-1}
\int_{\mathbb T^d}
\mathrm e^{\mathrm in\theta}
(2\cos\theta_1)^{\alpha_1}\cdots(2\cos\theta_d)^{\alpha_d}
\,\mathrm d\theta.
\end{align*}
The last integral factorizes into products of the following integrals:
\begin{align*}
\int_{\mathbb T}
\mathrm e^{\mathrm in_j\theta_j}(2\cos\theta_j)^{\alpha_j}\,\mathrm d\theta_j
&=
\sum_{k=0}^{\alpha_j}\frac{\alpha_j!}{k!(\alpha_j-k)!}\int_{\mathbb T}
\mathrm e^{\mathrm i(2k-\alpha_j+|n_j|)\theta_j}\,\mathrm d\theta_j,
\quad 
j=1,\dots,d.
\end{align*}
Each of the above integrals does not vanish only if $\alpha_j-|n_j|$ is non-negative and even, for which we have 
\begin{align*}
\int_{\mathbb T}
\mathrm e^{\mathrm in_j\theta_j}(2\cos\theta_j)^{\alpha_j}\,\mathrm d\theta_j
=
\frac{2\pi \alpha_j!}{\left(\frac{\alpha_j-|n_j|}2\right)!\left(\frac{\alpha_j+|n_j|}2\right)!}.
\end{align*}
Then we can proceed as 
\begin{align*}
G(z,n)
&
=
\sum_{\alpha_j-|n_j|\in2\mathbb N_0}
\frac{|\alpha|!}{
\prod_{j=1}^d\left[\left(\frac{\alpha_j-|n_j|}2\right)!\left(\frac{\alpha_j+|n_j|}2\right)!\right]}
(2d-z)^{-|\alpha|-1}
,
\end{align*}
and hence the first expression of the assertion follows by changing 
the above dummy indices $\alpha_j$ to $\beta_j$ through $\alpha_j=|n_j|+2\beta_j$.
We are done.
\end{proof}

\begin{remark*}
If we use the formula \eqref{1907152} below, the proof is much simpler.  
In fact, substitute \eqref{190804} into \eqref{1907152}, and we have 
\begin{align*}
G(z,n)
&
=
\sum_{\alpha\in\mathbb N_0^d}
\frac{1}{\alpha!\prod_{j=1}^d\bigl[(|n_j|+\alpha_j)!\bigr]}
\int_0^\infty 
\mathrm e^{-t(2d-z)}
t^{2|\alpha|+|n|}
\,\mathrm dt
.
\end{align*}
Then the assertion follows by the integral representation of the gamma function. 
\end{remark*}

For $d=2$ we can reduce from $F_C^{(2)}$ to ${_4F_3}$ as follows, thanks to a binomial-type theorem. 
We have not found a corresponding reduction formula for $d\ge 3$. 

\begin{corollary}\label{1908056}
If $d=2$, then 
\begin{align*}
G(z,n)
&
=
\sum_{k=0}^\infty 
\frac{ ((2k+|n|)!)^2}{(k+|n_1|)!(k+|n_2|)!(|n|+k)!k!}(4-z)^{-2k-|n|-1}
\\&
=
\frac{|n|!}{|n_1|!|n_2|!}
(4-z)^{-|n|-1}
{_4F_3}\left(
\genfrac{}{}{0pt}{}{\frac{|n|+1}2,\frac{|n|+1}2,\frac{|n|+2}2,\frac{|n|+2}2}{|n_1|+1,|n_2|+1,|n|+1}
;\frac{16}{(z-4)^2}\right)
.
\end{align*}
\end{corollary}

\begin{proof}
Due to Theorem~\ref{181110} it suffices to show the identity
\begin{align*}
\sum_{|\alpha|=k}
\frac{ (2|\alpha|+|n|)!}{\alpha_1!\alpha_2!(\alpha_1+|n_1|)!(\alpha_2+|n_2|)!}
=
\frac{ ((2k+|n|)!)^2}{(k+|n_1|)!(k+|n_2|)!(|n|+k)!k!}.
\end{align*}
We can verify this identity by comparing the $k$-th order terms of 
both sides of 
\begin{align*}
(1+t)^{k+|n_1|}(1+t)^{k+|n_2|}
=
(1+t)^{2k+|n|}.
\end{align*}
Let us omit the rest of the arguments.
\end{proof}

\subsection{Proof for endpoint thresholds in dimension one}\label{18111121}

Here we prove Theorem~\ref{190716}. 
Let us employ the following formula valid in all dimensions.

\begin{proposition}\label{190715}
For any $\mathop{\mathrm{Re}}z<0$ and $n\in\mathbb Z^d$ one has an expression
\begin{align}
G(z,n)
&
=\int_0^\infty 
\mathrm e^{-t(2d-z)}
\prod_{j=1}^d I_{n_j}(2t)
\,\mathrm dt.
\label{1907152}
\end{align}
\end{proposition}

\begin{remarks*}
\begin{enumerate}
\item
The above integration is the \emph{Laplace transform} $\mathcal L$, defined formally as 
$$\{\mathcal Lf\}(s)=\int_0^\infty \mathrm e^{-st}f(t)\,\mathrm dt.$$
\item
The \emph{modified Bessel functions of the first kind}, $I_\nu$, are defined as 
\begin{align}
I_\nu(w)=\mathrm i^{-\nu}J_\nu(\mathrm iw)=\sum_{k=0}^\infty\frac{(w/2)^{2k+\nu}}{k!\Gamma(k+\nu+1)},
\label{190804}
\end{align}
where $J_\nu$ are the Bessel functions of the first kind. 
The convergence of the integral \eqref{1907152} is guaranteed by
the following asymptotic formula \cite[(10.40.1)]{NIST:DLMF}: 
For any fixed $\nu\in\mathbb C$ and $\delta\in (0,\pi/2]$ 
\begin{align*}
I_\nu(w)\sim \frac{\mathrm e^w}{(2\pi w)^{1/2}}\sum_{k=0}^\infty (-1)^k\frac{a_k(\nu)}{w^k}
\ \ \text{as }w\to\infty\text{ with }|\arg w|\le \pi/2-\delta,
\end{align*}
where 
\begin{align*}
a_0(\nu)=1,\quad 
a_k(\nu)=\frac{(4\nu^2-1^2)(4\nu^2-3^2)\cdots(4\nu^2-(2k-1)^2)}{k!8^k}.
\end{align*}
\end{enumerate}
\end{remarks*}

\begin{proof}
Since  both sides of \eqref{1907152} are analytic in $\mathop{\mathrm{Re}}z<0$,
it suffices to prove the assertion only for real and negative $z<0$. Let $\Theta(\theta)=2d-2\sum_{j=1}^d\cos\theta_j$ be the symbol of $H_0$.
Then rewrite as follows
\begin{align*}
G(z,n)
&
=(2\pi)^{-d}\int_{\mathbb T^d}
\biggl(
\int_0^\infty \mathrm e^{-t(\Theta(\theta)-z)}\,\mathrm dt
\biggr)
\mathrm e^{\mathrm in\theta}
\,\mathrm d\theta
\\&
=(2\pi)^{-d}\int_0^\infty 
\mathrm e^{-t(2d-z)}
\Biggl(
\prod_{j=1}^d
\int_{\mathbb T}\mathrm e^{\mathrm in_j\theta_j}\mathrm e^{2t\cos \theta_j}
\,\mathrm d\theta_j\Biggr)
\,\mathrm dt.
\end{align*}
We quote the \emph{Jacobi--Anger expansion} \cite[(10.35.2)]{NIST:DLMF}: 
For $w,\theta\in\mathbb C$
$$\mathrm e^{w\cos \theta}=\sum _{k=-\infty }^{\infty}I_k(w)\mathrm e^{\mathrm ik\theta}.$$
Then we have 
\begin{align*}
(2\pi)^{-1}\int_{\mathbb T}\mathrm e^{\mathrm in_j\theta_j}\mathrm e^{2t\cos \theta_j}\,\mathrm d\theta_j
&=
I_{-n_j}(2t)
=I_{n_j}(2t)
,
\end{align*}
so that \eqref{1907152} follows. Hence we are done.
\end{proof}

\begin{proof}[Proof of Theorem~\ref{190716}]
Now we let $d=1$. The Laplace transform of the modified Bessel function
is well known. For any $\nu\in \mathbb N_0$ 
and $\mathop{\mathrm{Re}}s>|\omega|$
\begin{align}
\{\mathcal LI_\nu(\omega t)\}(s)
=\frac{\left(\sqrt{s+\omega}-\sqrt{s-\omega}\right)^{2\nu}}{(2\omega)^\nu\sqrt{(s-\omega)(s+\omega)}}.
\label{19071523}
\end{align}
Hence the formula \eqref{190807}
is a direct consequence of Theorem~\ref{190715} and \eqref{19071523}.

To prove \eqref{19083117} we use Theorem~\ref{181110}, \eqref{191221}, \cite[(15.8.5)]{NIST:DLMF}, 
\cite[(15.1.2)]{NIST:DLMF} and \cite[(5.5.5)]{NIST:DLMF}. 
Here and below we remark that  
$F(a,b;c;z)$ and $\mathbf F(a,b;c;z)$ are distinguished functions in this reference, 
see \cite[(15.1.2)]{NIST:DLMF}.  	
Then we can write for $\mathop{\mathrm{Re}}(2-z)>0$ and $\mathop{\mathrm{Im}}(2-z)>0$
\begin{align}
\begin{split}
G(z,n)
&
=
-
\frac{|n|}{2}
F\left(\genfrac{}{}{0pt}{}{\frac{1+n}2,\frac{1-n}2}{\frac32};\frac{z(4-z)}4\right)
\\&\phantom{{}={}}{}
+\frac{2-z}{2\sqrt{-z(4-z)}}
F\left(\genfrac{}{}{0pt}{}{\frac{1+n}2,\frac{1-n}2}{\frac12};\frac{z(4-z)}4\right)
.
\end{split}
\label{1908311811}
\end{align}
Next we rewrite it using \cite[(15.8.1)]{NIST:DLMF} with \cite[(15.1.2)]{NIST:DLMF},
and obtain 
\begin{align*}
\begin{split}
G(z,n)
&
=
-
\frac{|n|}{2}
F\left(\genfrac{}{}{0pt}{}{\frac{1+n}2,\frac{1-n}2}{\frac32};\frac{z(4-z)}4\right)
+\frac1{\sqrt{-z(4-z)}}
F\left(\genfrac{}{}{0pt}{}{\frac{n}2,-\frac{n}2}{\frac12};\frac{z(4-z)}4\right)
,
\end{split}
\end{align*}
Finally by \cite[(15.8.18)]{NIST:DLMF} and \cite[(15.8.20)]{NIST:DLMF} 
we obtain \eqref{19083117}.

We can prove \eqref{19083118} similarly to \eqref{19083117}.
We note that here 
an extra factor $(-1)^{n+1}$ appears in the counterpart of \eqref{1908311811}
due to $\mathop{\mathrm{Re}}(z-2)>0$, $\mathop{\mathrm{Im}}(z-2)>0$ 
and our choice of a branch of the fractional power, or the logarithm. 
We omit the rest of the arguments.
\end{proof}

\begin{remark*}
We can prove \eqref{19083117} also by using 
\cite[(2.2),(2.3)]{IJ1}, Chebyshev polynomials, 
their expressions in terms of hypergeometric functions,
and transformation formulas for hypergeometric functions. 
We omit the arguments.
\end{remark*}

\subsection{Diagonal values in dimension two}

Before the proofs of the main results in dimension two we compute the 
diagonal values of $G(z,n)$. The restriction to the diagonal further reduces 
${_4F_3}$ from Corollary~\ref{1908056} to $F={_2F_1}$,
for which plenty of transformation formulas are available. 
Based on the results obtained here,
we will proceed to the off-diagonal 
values in the following sections by exploiting the symmetries 
\begin{align}
G(z,n_1,n_2)=G(z,|n_1|,|n_2|)=G(z,n_2,n_1)
\label{1909039}
\end{align}
and the Helmholtz equation 
\begin{align}
\begin{split}
\delta_0[n_1]
\delta_0[n_2]
&=
(4-z)G(z,n_1,n_2)
-G(z,n_1+1,n_2)
-G(z,n_1-1,n_2)
\\&\phantom{{}={}}{}
-G(z,n_1,n_2+1)
-G(z,n_1,n_2-1).
\end{split}
\label{1909038}
\end{align}

Let us recall the notation $\mathcal P_0(m)=(-1)^{m}G(z,m,m)$ from \eqref{19122923}. 

\begin{proposition}\label{191230}
Let $d=2$. 
Then for any $|z-4|<4$ with $\mathop{\mathrm{Im}}z>0$ and $m\in\mathbb Z$
\begin{align*}
\mathcal P_0(m)
&
=
\frac{\mathrm i}{4\pi}
\sum_{k=0}^\infty
\frac{\left(\frac12+m\right)_k\left(\frac12-m\right)_k}{k!k!}
\left(\frac{z-4}{16}\right)^{2k}
\\&\phantom{{}=}
\cdot
\left[
2\psi(1+k)
-\psi\left(\frac12+m+k\right)
-\psi\left(\frac12-m+k\right)
-\log\left(-\frac{(z-4)^2}{16}\right)
\right],
\intertext{
and for any $|z(8-z)|<16$ with $\pm\mathop{\mathrm{Re}}(4-z)>0$ and $m\in\mathbb Z$
}
\mathcal P_0(m)
&
=
\pm\frac{(-1)^m}{4\pi}
\sum_{k=0}^\infty
\frac{\left(\frac12+m\right)_k\left(\frac12-m\right)_k}{k!k!}
\left(\frac{z(8-z)}{16}\right)^k
\\&\phantom{{}={}}{}
\cdot
\left[
2\psi(1+k)
-\psi\left(\frac12+m+k\right)
-\psi\left(\frac12-m+k\right)
-\log\frac{z(z-8)}{16}
\right]
,
\end{align*}
respectively.
\end{proposition}

\begin{proof}
By \eqref{1909039} we obviously have $\mathcal P_0(m)=\mathcal P_0(-m)$,
hence we may assume $m\ge 0$. 
On the diagonal the second formula from Corollary~\ref{1908056} reduces to 
\begin{align}
\begin{split}
\mathcal P_0(m)
&
=
(-1)^m
\frac{(2m)!}{m!m!}
(4-z)^{-2m-1}
{_2F_1}\biggl(\genfrac{}{}{0pt}{}{\frac12+m,\frac12+m}{1+2m};\frac{16}{(z-4)^2}\biggr)
.
\end{split}
\label{1908071742}
\end{align}
Now we use the formula \cite[(15.8.8)]{NIST:DLMF}
along with \cite[(5.1.2)]{NIST:DLMF}.
Noting the branch of fractional power, 
we can rewrite \eqref{1908071742} as 
\begin{align*}
\begin{split}
\mathcal P_0(m)
&
=
\frac{\mathrm i(2m)!(2m)!}{2^{2+4m}m!m!\Gamma\left(\frac12+m\right)}
\sum_{k=0}^\infty
\frac{(-1)^k(\frac12+m)_k}{k!k!\Gamma(\frac12+m-k)}\left(\frac{z-4}{4}\right)^{2k}
\\&\phantom{{}=}{}\cdot
\left[2\psi(1+k)-\psi\left(\frac12+m+k\right)-\psi\left(\frac12+m-k\right)
-\log\left(-\frac{(z-4)^2}{16}\right)\right]
.
\end{split}
\end{align*}
Then by \cite[(5.5.3)]{NIST:DLMF}, \cite[(5.5.4)]{NIST:DLMF} and \cite[(5.5.5)]{NIST:DLMF}
the first formula of the assertion follows. 
Analogously, we use the formula \cite[(15.8.11)]{NIST:DLMF}
along with \cite[(5.1.2)]{NIST:DLMF}, 
and rewrite \eqref{1908071742} as 
\begin{align*}
\mathcal P_0(m)
&
=
\pm\frac{(-1)^m(2m)!(2m)!}{2^{2+4m}m!m!\Gamma\left(\frac12+m\right)}
\sum_{k=0}^\infty\frac{(-1)^k\left(\frac12+m\right)_k}{k!k!\Gamma\left(\frac12+m-k\right)}
\left(\frac{z(8-z)}{16}\right)^k
\\&\phantom{{}={}}{}
\cdot
\biggl[
2\psi(1+k)
-\psi\left(\frac12+m+k\right)
-\psi\left(\frac12+m-k\right)
-\log\left(\frac{z(z-8)}{16}\right)
\biggr]
.
\end{align*}
Then by \cite[(5.5.3)]{NIST:DLMF}, \cite[(5.5.4)]{NIST:DLMF} and \cite[(5.5.5)]{NIST:DLMF}
we obtain the second formula of the assertion. Hence we are done.
\end{proof}

\subsection{Proof for embedded threshold in dimension two}\label{19122622}

Since we already have an explicit candidate for the expansion around 
the embedded threshold in dimension $d=2$, 
it suffices to check that the candidate actually satisfies 
\eqref{1909039} and \eqref{1909038} with the same diagonal values 
as in Proposition~\ref{191230}.
Note that a function satisfying these properties is unique.

\begin{proof}[Proof of Theorem~\ref{19083122}] 
It is clear that the right-hand sides of the asserted identities satisfy 
\eqref{1909039}, and have the same diagonal values as those in Proposition~\ref{191230}. 
Hence it suffices to show that they also satisfy \eqref{1909038}. 
More precisely, we actually claim that the first terms on the right-hand sides of the assertion 
form a fundamental solution to $H_0-z$,
and the second a generalized eigenfunction. 

We first consider the first terms of the asserted identities. 
Set their coefficients as, for $k\in \mathbb N_0$ and $n\in\mathbb Z^2$, 
\begin{align*}
a_k[n]
&=(-1)^{\max\{|n_1|,|n_2|\}}
\frac{
\frac{|n_1+n_2|}2\frac{|n_1-n_2|}2
\left(\frac{2+n_1+n_2}2\right)_k
\left(\frac{2+n_1-n_2}2\right)_k
\left(\frac{2-n_1+n_2}2\right)_k
\left(\frac{2-n_1-n_2}2\right)_k
}{k!k!\left(\frac32\right)_k\left(\frac32\right)_k}
,\\
b_k[n]
&=
\frac{(-1)^{\max\{|n_1|,|n_2|\}}}4
\frac{
\left(\frac{1+n_1+n_2}2\right)_k
\left(\frac{1+n_1-n_2}2\right)_k
\left(\frac{1-n_1+n_2}2\right)_k
\left(\frac{1-n_1-n_2}2\right)_k
}{k!k!\left(\frac12\right)_k\left(\frac12\right)_k}. 
\end{align*}
Then the proof reduces to show that for $|n|$ even and $k\in \mathbb N$
\begin{align}
b_0[n+e_1]+b_0[n-e_1]+b_0[n+e_2]+b_0[n-e_2]
&=-\delta_{n_1}[n_1]\delta_{n_2}[n_2]
,
\label{20010215}
\\
b_k[n+e_1]+b_k[n-e_1]+b_k[n+e_2]+b_k[n-e_2]
&=-4a_{k-1}[n],
\label{20010216}
\intertext{and that for $|n|$ odd and $k\in \mathbb N_0$}
a_k[n+e_1]+a_k[n-e_1]+a_k[n+e_2]+a_k[n-e_2]
&=-4b_{k}[n].
\label{20010217}
\end{align}
The identity \eqref{20010215} is clear, since 
\begin{align*}
b_0[n]
&=\frac{(-1)^{\max\{|n_1|,|n_2|\}}}4. 
\end{align*}
For the proof of \eqref{20010216} we may assume $n_1\ge n_2\ge 0$ due to the symmetry \eqref{1909039},
but if $n_1=n_2\in\mathbb N_0$ and $k\in\mathbb N$,
then \eqref{20010216} is again clear since 
all the terms there are $0$. 
Thus we may assume $n_1>n_2\ge 0$ with $|n|$ even and $k\in\mathbb N$, and then 
\begin{align*}
&b_k[n+e_1]+b_k[n-e_1]+b_k[n+e_2]+b_k[n-e_2]
\\&
=
\frac{(-1)^{n_1}}{k!k!\left(\frac32\right)_{k-1}\left(\frac32\right)_{k-1}}
\\&\phantom{{}={}}{}\cdot
\biggl[
-\left(\frac{2+n_1+n_2}2\right)_k
\left(\frac{2+n_1-n_2}2\right)_k
\left(\frac{-n_1+n_2}2\right)_k
\left(\frac{-n_1-n_2}2\right)_k
\\&\phantom{{}={}\cdot\biggl[}
-\left(\frac{n_1+n_2}2\right)_k
\left(\frac{n_1-n_2}2\right)_k
\left(\frac{2-n_1+n_2}2\right)_k
\left(\frac{2-n_1-n_2}2\right)_k
\\&\phantom{{}={}\cdot\biggl[}
+
\left(\frac{2+n_1+n_2}2\right)_k
\left(\frac{n_1-n_2}2\right)_k
\left(\frac{2-n_1+n_2}2\right)_k
\left(\frac{-n_1-n_2}2\right)_k
\\&\phantom{{}={}\cdot\biggl[}
+
\left(\frac{n_1+n_2}2\right)_k
\left(\frac{2+n_1-n_2}2\right)_k
\left(\frac{-n_1+n_2}2\right)_k
\left(\frac{2-n_1-n_2}2\right)_k
\biggr]
\\&
=
-4
\frac{(-1)^{n_1}}{(k-1)!(k-1)!\left(\frac32\right)_{k-1}\left(\frac32\right)_{k-1}}
\frac{n_1+n_2}2\frac{n_1-n_2}2
\\&\phantom{{}={}}\cdot
\left(\frac{2+n_1+n_2}2\right)_{k-1}
\left(\frac{2+n_1-n_2}2\right)_{k-1}
\left(\frac{2-n_1+n_2}2\right)_{k-1}
\left(\frac{2-n_1-n_2}2\right)_{k-1}
\\&
=-4a_{k-1}[n].
\end{align*}
This proves \eqref{20010216}. 
We can proceed more or less similarly for \eqref{20010217},
and we omit it. 
Hence we are done with the first terms of the asserted identities.

To prove the claim concerning the second terms of the asserted identities we further set 
for $k\in\mathbb N_0$ and $n\in\mathbb Z^2$
\begin{align*}
\widetilde a_k[n]
&=(-1)^{\max\{|n_1|,|n_2|\}}
\frac{
\frac{|n_1+n_2|}2\frac{|n_1-n_2|}2
\left(\frac{2+n_1+n_2}2\right)_k
\left(\frac{2+n_1-n_2}2\right)_k
\left(\frac{2-n_1+n_2}2\right)_k
\left(\frac{2-n_1-n_2}2\right)_k
}{k!k!\left(\frac32\right)_k\left(\frac32\right)_k}
\\&\phantom{{}={}}\cdot
\biggl[
2\psi\left(1+k\right)
+2\psi\left(\frac32+k\right)
\\&\phantom{{}={}}\phantom{{}\cdot\biggl[} {}
-\psi\left(\frac{2+n_1+n_2}2+k\right)
-\psi\left(\frac{2+n_1-n_2}2+k\right)
\\&\phantom{{}={}}\phantom{{}\cdot\biggl[} {}
-\psi\left(\frac{2-n_1+n_2}2+k\right)
-\psi\left(\frac{2-n_1-n_2}2+k\right)
\biggr]
,\\
\widetilde b_k[n]
&=
\frac{(-1)^{\max\{|n_1|,|n_2|\}}}4
\frac{
\left(\frac{1+n_1+n_2}2\right)_k
\left(\frac{1+n_1-n_2}2\right)_k
\left(\frac{1-n_1+n_2}2\right)_k
\left(\frac{1-n_1-n_2}2\right)_k
}{k!k!\left(\frac12\right)_k\left(\frac12\right)_k}
\\&\phantom{{}={}}\cdot
\biggl[
2\psi\left(1+k\right)
+2\psi\left(\frac12+k\right)
\\&\phantom{{}={}}\phantom{{}\cdot\biggl[} {}
-\psi\left(\frac{1+n_1+n_2}2+k\right)
-\psi\left(\frac{1+n_1-n_2}2+k\right)
\\&\phantom{{}={}}\phantom{{}\cdot\biggl[} {}
-\psi\left(\frac{1-n_1+n_2}2+k\right)
-\psi\left(\frac{1-n_1-n_2}2+k\right)
\biggr]
.
\end{align*}
Then the proof reduces to show that for $|n|$ even and $k\in \mathbb N_0$
\begin{align}
a_k[n+e_1]+a_k[n-e_1]+a_k[n+e_2]+a_k[n-e_2]
&=-4b_k[n],
\label{202001022230}
\\
\widetilde a_k[n+e_1]+\widetilde a_k[n-e_1]+\widetilde a_k[n+e_2]+\widetilde a_k[n-e_2]
&=-4\widetilde b_k[n],
\label{202001022231}
\intertext{and that for $|n|$ odd and $k\in \mathbb N$}
b_0[n+e_1]+b_0[n-e_1]+b_0[n+e_2]+b_0[n-e_2]
&=0
,
\label{202001022232}
\\
\widetilde b_0[n+e_1]+\widetilde b_0[n-e_1]+\widetilde b_0[n+e_2]+\widetilde b_0[n-e_2]
&=0
,
\label{202001022233}
\\
b_k[n+e_1]+b_k[n-e_1]+b_k[n+e_2]+b_k[n-e_2]
&=-4a_{k-1}[n]
,
\label{202001022234}
\\
\widetilde b_k[n+e_1]+\widetilde b_k[n-e_1]+\widetilde b_k[n+e_2]+\widetilde b_k[n-e_2]
&=-4\widetilde a_{k-1}[n]
.
\label{202001022235}
\end{align}
By the symmetry \eqref{1909039} it suffices to prove \eqref{202001022230}--\eqref{202001022235}
only for $n_1\ge n_2\ge 0$.
The proofs of \eqref{202001022230}, \eqref{202001022232} and \eqref{202001022233}
are similar to those of \eqref{20010215}--\eqref{20010217},
and we omit them. 
(If we think of \eqref{20010215}--\eqref{20010217} as polynomials in $n$, 
then \eqref{202001022230}, \eqref{202001022232} and \eqref{202001022233}
for $n_1>n_2\ge 0$ follow from \eqref{20010215}--\eqref{20010217}.)
The proofs of \eqref{202001022231}, \eqref{202001022234} and \eqref{202001022235}
are a bit more complicated, but routine. 
We omit the detail again, and hence we are done. 
\end{proof}

\begin{remark*}
To prove Theorem~\ref{19083122} it seems that 
we could apply the formulas \cite[(2.1) and (2.2)]{SS} to Corollary~\ref{1908056},
but then the resulting formulas miss the first terms on the right-hand sides of the asserted identities. 
Without them we cannot derive a correct fundamental solution, and this is a contradiction.
So far we could not identify where our application of \cite{SS} has been wrong,
while it still helped us to guess a correct candidate. 
\end{remark*}

\subsection{Proof for endpoint thresholds in dimension two}\label{19122623}

We do not have a simple candidate for the expansion around 
the endpoint thresholds unlike the embedded one,
and here we proceed to solve the equation \eqref{1909038}
with symmetry \eqref{1909039} and the diagonal values from Proposition~\ref{191230}.
For that we first derive from \eqref{1909038}
recurrence relations for $\mathcal P(m,l)$ and $\mathcal Q(m,l)$. 
This will be done essentially by taking double summation of \eqref{1909038}.

\begin{proposition}\label{191010}
Let $d=2$. 
For any $z\in\mathbb C\setminus [0,8]$ and $(m,l)\in\mathbb Z^2$ 
\begin{align}
\begin{split}
\mathcal P(m,l)
&
=
\frac14|m||l|(z-4)
-\mathcal P_0(0)
+\mathcal P_0(m)
+\mathcal P_0(l)
\\&\phantom{{}={}}{}
+\frac14(z-4)^2\sum_{|\mu|\le |m|}\sum_{|\nu|\le |l|}(|m|-|\mu|)(|l|-|\nu|)\mathcal P(\mu,\nu),
\end{split}
\label{19101420}
\intertext{and for any $z\in\mathbb C\setminus [0,8]$ and $(m,l)\in\mathbb N_0^2$}
\begin{split}
\mathcal Q(m,l)
&
=
-\frac14
-\frac{z-4}4\sum_{|\mu|\le m}\sum_{|\nu|\le l}\mathcal P(\mu,\nu)
.
\end{split}
\label{19101421}
\end{align}
\end{proposition}
\begin{proof}
Due to the symmetry \eqref{1909039} it suffices to prove the assertion 
for $(m,l)\in\mathbb N_0^2$.

\smallskip
\noindent
\textit{Step 1. } 
We first show 
\begin{align*}
\mathcal Q(0,0)
=
-\frac14-\frac{z-4}4\mathcal P(0,0)
.
\end{align*}
However, it is in fact straightforward
by letting $n=(0,0)$ in \eqref{1909038} and using \eqref{1909039}.

\smallskip
\noindent
\textit{Step 2. } 
We next claim that for any $m\in\mathbb N$
\begin{align*}
\begin{split}
\mathcal Q(m,0)
=
-\frac14
-\frac{z-4}4\sum_{|\mu|\le m}\mathcal P(\mu,0)
.
\end{split}
\end{align*}
In order to prove this we let $n=(\mu,\mu)$ with $\mu\in\mathbb N$ in \eqref{1909038}, 
and then by \eqref{1909039}
\begin{align}
0=
(4-z)\mathcal P(\mu,0)
-2\mathcal Q(\mu,0)
+2\mathcal Q(\mu-1,0)
.
\label{19090310}
\end{align}
Summing up \eqref{19090310} in $\mu$, we obtain 
\begin{align}
0=
(4-z)\sum_{\mu=1}^m\mathcal P(\mu,0)
-2\mathcal Q(m,0)
+2\mathcal Q(0,0)
.
\label{19090321}
\end{align}
Now by \eqref{19090321}, Step 1, and the identity $\mathcal P(\mu,0)=\mathcal P(-\mu,0)$, 
the claim follows.

\smallskip
\noindent
\textit{Step 3. }
We next claim that for any $l\in\mathbb N$
\begin{align*}
\begin{split}
\mathcal Q(0,l)
=
-\frac14
-\frac{z-4}4\sum_{|\nu|\le l}\mathcal P(0,\nu)
,
\end{split}
\end{align*}
but we can prove it similarly to Step 2,
and we omit the detail. 
Note that we can also prove it by Step 2 and the symmetry 
\begin{align*}
\mathcal Q(0,\nu)=G(z,\nu+1,-\nu)=G(z,\nu+1,\nu)=\mathcal Q(\nu,0).
\end{align*}

\smallskip
\noindent
\textit{Step 4. }
Now we prove \eqref{19101421} in the general case. 
It suffices to consider only $(m,l)\in\mathbb N^2$ due to Steps 1--3.
By letting $n=(\mu+\nu,\mu-\nu)$ with $\mu,\nu\in\mathbb N$ in \eqref{1909038}
\begin{align*}
0&=
\mathcal Q(\mu,\nu)
-\mathcal Q(\mu-1,\nu)
-\mathcal Q(\mu,\nu-1)
+\mathcal Q(\mu-1,\nu-1)
+(z-4)\mathcal P(\mu,\nu).
\end{align*}
Then we take a double summation in $\nu$ and $\mu$ to obtain 
\begin{align*}
0
&
=
\mathcal Q(m,l)
-\mathcal Q(0,l)
-\mathcal Q(m,0)
+\mathcal Q(0,0)
+(z-4)\sum_{\mu=1}^m\sum_{\nu=1}^l\mathcal P(\mu,\nu).
\end{align*}
Now we substitute the result of Steps 1--3, and then \eqref{19101421} follows.

\smallskip
\noindent
\textit{Step 5. }
Here we prove that for any $(m,l)\in\mathbb N_0^2$ 
\begin{align*}
\mathcal P(m,l)
&
=
-\mathcal P_0(0)
+\mathcal P_0(m)
+\mathcal P_0(l)
-(z-4)\sum_{\mu=0}^{m-1}\sum_{\nu=0}^{l-1}\mathcal Q(\mu,\nu).
\end{align*}
Similarly to the arguments so far, 
by letting $n=(\mu+\nu+1,\mu-\nu)$ in \eqref{1909038}
\begin{align*}
0&=
\mathcal P(\mu+1,\nu+1)
-\mathcal P(\mu+1,\nu)
+\mathcal P(\mu,\nu)
-\mathcal P(\mu,\nu+1)
+(z-4)\mathcal Q(\mu,\nu)
.
\end{align*}
Then we take the double summation in $\nu$ and $\mu$, 
and the claim follows.

\smallskip
\noindent
\textit{Step 6. }
Finally we prove \eqref{19101420},
but it is rather straightforward to 
obtain \eqref{19101420}
after substituting \eqref{19101421} into the result of Step 5.
We omit the detail. 
Hence we are done.
\end{proof}

To solve the recurrence relations from Proposition~\ref{191010}
it would be useful to have the following simple summation formula.

\begin{lemma}\label{191017}
Let $p,q\in\mathbb Z$ with $p\le q$, $r\in \mathbb C$ and $k\in\mathbb N_0$. Then 
\begin{align*}
\sum_{j=p}^q(j+r)_k=\frac1{k+1}\bigl[(q+r)_{k+1}-(p+r-1)_{k+1}\bigr]
.
\end{align*}
\end{lemma}
\begin{proof}
Assume $k\neq 0$, and then we have 
\begin{align*}
(k+1)(j+r)_k
&=
(j+r)_{k}(j+r+k)-(j+r)_{k}(j+r-1)
\\&=
(j+r)_{k+1}-(j+r-1)_{k+1}.
\end{align*}
The resulting identity is in fact true also for $k=0$.
Hence the assertion follows by the method of differences.
\end{proof}

The first identity from Theorem~\ref{19080717} was already proved in 
Proposition~\ref{191230}. 
Now we prove the second and third ones by solving the recurrence relations from Proposition~\ref{191010}.

\begin{proof}[Proof of the second identity from Theorem~\ref{19080717}]
\textit{Step 1. } 
First we solve the equation \eqref{19101420} from Proposition~\ref{191010}. 
Let us define an operator $\mathcal T$ 
acting on a general sequence $\mathcal F\colon\mathbb Z^2\to\mathbb C$ as 
$$(\mathcal T\mathcal F)(m,l)
=\frac14(z-4)^2\sum_{|\mu|\le |m|}\sum_{|\nu|\le |l|}
\bigl(|m|-|\mu|\bigr)\bigl(|l|-|\nu|\bigr) \mathcal F(\mu,\nu).$$
Then we can expect a formal solution to \eqref{19101420} of the form 
\begin{align}
\mathcal P(m,l)
&
=
\sum_{k=0}^\infty
\left[\mathcal T^k
\left(\frac14|\mu||\nu|(z-4)
-\mathcal P_0(0)
+\mathcal P_0(\mu)
+\mathcal P_0(\nu)
\right)\right](m,l)
.
\label{191015}
\end{align}
Here let us verify the last formula. 
For that it suffices to show that for each 
$\mathcal F\colon\mathbb Z^2\to\mathbb C$ and 
$(m,n)\in\mathbb Z^2$ there exists $k_0\in\mathbb N_0$ such that 
for any $k\ge k_0$
$$(\mathcal T^k\mathcal F)(m,n)=0,$$
which implies that the series \eqref{191015} in fact is a pointwise finite sum. 
We can actually prove a more explicit formula: 
For any $k\in \mathbb N$
\begin{align}
\begin{split}
(\mathcal T^k\mathcal F)(m,l)
&
=
\frac{(z-4)^{2k}}{4(2k-1)!(2k-1)!}
\sum_{|\mu|\le |m|}
\sum_{|\nu|\le |l|}
\\&\phantom{{}={}}{}
\bigl(|m|-|\mu|-k+1\bigr)_{2k-1}
\bigl(|l|-|\nu|-k+1\bigr)_{2k-1}
\mathcal F(\mu,\nu)
.
\end{split}
\label{19101020}
\end{align}

Now we prove \eqref{19101020} by induction. 
By definition of $\mathcal T$ the formula \eqref{19101020} is obviously true for $k=1$.
Assume \eqref{19101020} for some $k\in \mathbb N$, 
and then 
\begin{align*}
(\mathcal T^{k+1}F)(m,l)
&
=
\frac{-(z-4)^{2(k+1)}}{16(2k-1)!(2k-1)!}
\sum_{|a|\le |m|}
\sum_{|b|\le |l|}
\sum_{|\mu|\le |a|}
\sum_{|\nu|\le |b|}
\bigl(|m|-|a|\bigr)\bigl(|l|-|b|\bigr)
\\&\phantom{{}={}}{}
\cdot
\bigl(|a|-|\mu|-k+1\bigr)_{2k-1}
\bigl(|b|-|\nu|-k+1\bigr)_{2k-1}
F(\mu,\nu)
.
\end{align*}
We can compute the summation with respect to $a$ using Lemma~\ref{191017} as 
\begin{align*}
&
\sum_{|a|\le |m|}
\sum_{|\mu|\le |a|}
\bigl(|m|-|a|\bigr)\bigl(|a|-|\mu|-k+1\bigr)_{2k-1}
\\&=
2\sum_{|\mu|\le |m|}
\sum_{|\mu|\le a\le |m|}
\bigl(|m|-a\bigr)\bigl(a-|\mu|-k+1\bigr)_{2k-1}
\\&=
2\sum_{|\mu|\le |m|}
\sum_{|\mu|\le a\le |m|}
\bigl[
\bigl(|m|-|\mu|+k\bigr)\bigl(a-|\mu|-k+1\bigr)_{2k-1}-\bigl(a-|\mu|-k+1\bigr)_{2k}
\bigr]
\\&=
2\sum_{|\mu|\le |m|}
\biggl\{\frac1{2k}\bigl(|m|-|\mu|+k\bigr)
\bigl(|m|-|\mu|-k+1\bigr)_{2k}
\\&\phantom{{}={}}{}
\qquad\qquad
-
\frac1{2k+1}
\bigl(|m|-|\mu|-k+1\bigr)_{2k+1}
\biggl\}
\\&=
\frac2{(2k+1)(2k)}\sum_{|\mu|\le |m|}
\bigl(|m|-|\mu|-k\bigr)_{2k+1}
.
\end{align*}
We can similarly show that 
\begin{align*}
&
\sum_{|b|\le |l|}
\sum_{|\nu|\le |b|}
\bigl(|l|-|b|\bigr)\bigl(|b|-|\nu|-k+1\bigr)_{2k-1}
=
\frac2{(2k+1)(2k)}\sum_{|\nu|\le |l|}
\bigl(|l|-|\nu|-k\bigr)_{2k+1}
.
\end{align*}
Hence the claimed formula \eqref{19101020}, and hence \eqref{191015}, are verified.

\smallskip
\noindent
\textit{Step 2. } 
Here we prove the second formula from Theorem~\ref{19080717}.
By the formulas \eqref{191015} and \eqref{19101020} we can write  
\begin{align}
\begin{split}
\mathcal P(m,l)
&
=
\frac14|m||l|(z-4)
-\mathcal P_0(0)
+\mathcal P_0(m)
+\mathcal P_0(l)
\\&\phantom{{}={}}{}
+\sum_{k=1}^\infty
\frac{(4-z)^{2k}}{4(2k-1)!(2k-1)!}
\\&\phantom{{}={}+\sum}{}
\cdot
\sum_{|\mu|\le |m|}
\sum_{|\nu|\le |l|}
\bigl(|m|-|\mu|-k+1\bigr)_{2k-1}
\bigl(|l|-|\nu|-k+1\bigr)_{2k-1}
\\&\phantom{{}={}+\sum}{}
\cdot
\left(\frac14|\mu||\nu|(z-4)
-\mathcal P_0(0)
+\mathcal P_0(\mu)
+\mathcal P_0(\nu)
\right)
.
\end{split}
\label{1910152213}
\end{align}
Let us further compute \eqref{1910152213}. 
Actually it suffices to compute the following summations in $\mu$ and $\nu$: 
By letting $\mu'=|m|-\mu+1$ and using Lemma~\ref{191017}
\begin{align}
\begin{split}
&
\sum_{|\mu|\le |m|}\bigl(|m|-|\mu|-k+1\bigr)_{2k-1}|\mu|
\\&
=
2\sum_{\mu'=1}^{|m|}\bigl(\mu'-k\bigr)_{2k-1}(|m|-\mu'+1)
\\&
=
2\sum_{\mu'=1}^{|m|}\bigl[\bigl(|m|+k\bigr)\bigl(\mu'-k\bigr)_{2k-1}-\bigl(\mu'-k\bigr)_{2k}\bigr]
\\&
=
2
\biggl[
\frac1{2k}\bigl(|m|+k\bigr)\bigl(|m|-k\bigr)_{2k}
-\frac1{2k+1}\bigl(|m|-k\bigr)_{2k+1}\biggr]
\\&
=
\frac{2(-1)^k}{(2k+1)(2k)}|m|(1+m)_{k}(1-m)_{k}
,
\end{split}
\label{1910152214}
\end{align}
and, similarly,  
\begin{align}
\sum_{|\nu|\le |l|}\bigl(|l|-|\nu|-k+1\bigr)_{2k-1}|\nu|
=
\frac{2(-1)^k}{(2k+1)(2k)}|l|(1+l)_{k}(1-l)_{k}
.
\label{1910152215}
\end{align}
We also have by letting $\mu'=|m|-\mu+1$ and using Lemma~\ref{191017}
\begin{align}
\begin{split}
\sum_{|\mu|\le |m|}\bigl(|m|-|\mu|-k+1\bigr)_{2k-1}
&
=
\bigl(|m|-k+1\bigr)_{2k-1}
+2\sum_{\mu'=1}^{|m|}\bigl(\mu'-k\bigr)_{2k-1}
\\&
=
\bigl(|m|-k+1\bigr)_{2k-1}
+\frac2{(2k)}\bigl(|m|-k\bigr)_{2k}
\\&
=
\frac{2(-1)^k}{(2k)}(m)_{k}(-m)_{k}
,
\end{split}
\label{1910152216}
\end{align}
and, similarly, 
\begin{align}
&\sum_{|\nu|\le |l|}\bigl(|l|-|\nu|-k+1\bigr)_{2k-1}
=
\frac{2(-1)^k}{(2k)}(l)_{k}(-l)_{k}
.
\label{1910152217}
\end{align}
Now we substitute \eqref{1910152214}--\eqref{1910152217}
into \eqref{1910152213}. Then after some computations we obtain 
\begin{align*}
&
\mathcal P(m,l)
\\&
=
\frac14|m||l|\sum_{k=0}^\infty
\frac{(1+m)_{k}(1-m)_{k}(1+l)_{k}(1-l)_{k}(z-4)^{2k+1}}{(2k+1)!(2k+1)!}
\\&\phantom{{}={}}{}
+\mathcal P_0(0)\sum_{k=0}^\infty\frac{(m)_{k}(-m)_{k}(l)_{k}(-l)_{k}(z-4)^{2k}}{(2k)!(2k)!}
\\&\phantom{{}={}}{}
+\sum_{\mu=1}^{|m|}[\mathcal P_0(\mu)-\mathcal P_0(\mu-1)]
\sum_{k=0}^\infty\frac{(1+|m|-\mu)_{k}(\mu-|m|)_{k}(l)_{k}(-l)_{k}(z-4)^{2k}}{(2k)!(2k)!}
\\&\phantom{{}={}}{}
+\sum_{\nu=1}^{|l|}[\mathcal P_0(\nu)-\mathcal P_0(\nu-1)]
\sum_{k=0}^\infty\frac{(m)_{k}(-m)_{k}(1+|l|-\nu)_{k}(\nu-|l|)_{k}(z-4)^{2k}}{(2k)!(2k)!}
,
\end{align*}
hence the second formula from Theorem~\ref{19080717} follows.
\end{proof}

\begin{proof}[Proof of the third identity from Theorem~\ref{19080717}]
Substitute \eqref{1910152213} into the 
the second identity from Proposition~\ref{191010}
to get 
\begin{align*}
\mathcal Q(m,l)
&
=
-\frac14
-\frac{(z-4)^2}{16}m(m+1)l(l+1)
+\frac{z-4}4(2m+1)(2l+1)\mathcal P_0(0)
\\&\phantom{{}={}}{}
-\frac{z-4}4(2l+1)\sum_{|\mu|\le m}\mathcal P_0(\mu)
-\frac{z-4}4(2m+1)\sum_{|\nu|\le l}\mathcal P_0(\nu)
\\&\phantom{{}={}}{}
-\sum_{k=1}^\infty
\frac{(z-4)^{2k+1}}{16(2k-1)!(2k-1)!}
\\&\phantom{{}={}+}{}
\cdot
\sum_{|a|\le m}\sum_{|b|\le l}
\sum_{|\mu|\le |a|}
\sum_{|\nu|\le |b|}
\bigl(|a|-|\mu|-k+1\bigr)_{2k-1}
\bigl(|b|-|\nu|-k+1\bigr)_{2k-1}
\\&\phantom{{}={}+}{}
\cdot
\left(\frac14|\mu||\nu|(z-4)
-\mathcal P_0(0)
+\mathcal P_0(\mu)
+\mathcal P_0(\nu)
\right)
.
\end{align*}
Note that by Lemma~\ref{191017}
\begin{align*}
\sum_{|a|\le m}\sum_{|\mu|\le |a|}
\bigl(|a|-|\mu|-k+1\bigr)_{2k-1}
&=
2\sum_{|\mu|\le m}\sum_{|\mu|\le a\le m}
\bigl(a-|\mu|-k+1\bigr)_{2k-1}
\\&=
\frac2{(2k)}\sum_{|\mu|\le m}\bigl(m-|\mu|-k+1\bigr)_{2k}
,
\end{align*}
and, similarly, 
\begin{align*}
&
\sum_{|b|\le l}\sum_{|\nu|\le |b|}
\bigl(|b|-|\nu|-k+1\bigr)_{2k-1}
=
\frac2{(2k)}\sum_{|\nu|\le l}\bigl(l-|\nu|-k+1\bigr)_{2k}
.
\end{align*}
Then we can proceed as 
\begin{align*}
\mathcal Q(m,l)
&
=
-\frac14
-\frac{(z-4)^2}{16}m(m+1)l(l+1)
+\frac{z-4}4(2m+1)(2l+1)\mathcal P_0(0)
\\&\phantom{{}={}}{}
-\frac{z-4}4(2l+1)\sum_{|\mu|\le m}\mathcal P_0(\mu)
-\frac{z-4}4(2m+1)\sum_{|\nu|\le l}\mathcal P_0(\nu)
\\&\phantom{{}={}}{}
-\sum_{k=1}^\infty
\frac{(z-4)^{2k+1}}{4(2k)!(2k)!}
\sum_{|\mu|\le m}\sum_{|\nu|\le l}
\bigl(m-|\mu|-k+1\bigr)_{2k}
\bigl(l-|\nu|-k+1\bigr)_{2k}
\\&\phantom{{}={}+}{}
\cdot
\left(\frac14|\mu||\nu|(z-4)
-\mathcal P_0(0)
+\mathcal P_0(\mu)
+\mathcal P_0(\nu)
\right)
.
\end{align*}
Similarly to \eqref{1910152214}--\eqref{1910152217}
we can compute
\begin{align}
\begin{split}
\sum_{|\mu|\le m}\bigl(m-|\mu|-k+1\bigr)_{2k}|\mu|
&=
\frac{2(-1)^{k+1}}{(2k+2)(2k+1)}(1+m)_{k+1}(-m)_{k+1},
\end{split}
\label{1910152214b}
\\
\begin{split}
\sum_{|\nu|\le l}\bigl(l-|\nu|-k+1\bigr)_{2k}|\nu|
&=
\frac{2(-1)^{k+1}}{(2k+2)(2k+1)}(1+l)_{k+1}(-l)_{k+1},
\end{split}
\label{1910152215b}
\\
\begin{split}
\sum_{|\mu|\le m}\bigl(m-|\mu|-k+1\bigr)_{2k}
&
=
\frac{(-1)^k}{2k+1}(2m+1)(1+m)_k(-m)_k
,
\end{split}
\label{1910152216b}
\\
\begin{split}
\sum_{|\nu|\le l}\bigl(l-|\nu|-k+1\bigr)_{2k}
&
=
\frac{(-1)^k}{2k+1}(2l+1)(1+l)_k(-l)_k
.
\end{split}
\label{1910152217b}
\end{align}
Now substitute \eqref{1910152214b}--\eqref{1910152217b},
and then after some computations 
\begin{align*}
\mathcal Q(m,l)
&
=
-\sum_{k=0}^\infty
\frac{(1+m)_{k}(-m)_{k}(1+l)_{k}(-l)_{k}(z-4)^{2k}}{4(2k)!(2k)!}
\\&\phantom{{}={}}{}
+\mathcal P_0(0)(2m+1)(2l+1)
\sum_{k=0}^\infty
\frac{(1+m)_k(-m)_k(1+l)_k(-l)_k(z-4)^{2k+1}}{4(2k+1)!(2k+1)!}
\\&\phantom{{}={}}{}
-
(2l+1)\sum_{\mu=-m}^{m}\mathcal P_0(\mu)
\\&\phantom{{}={}-}{}\cdot
\sum_{k=0}^\infty
\frac{(1+m-|\mu|)_{k}(|\mu|-m)_{k}(1+l)_k(-l)_k(z-4)^{2k+1}}{4(2k+1)!(2k)!}
\\&\phantom{{}={}}{}
-
(2m+1)\sum_{\nu=-l}^{l}\mathcal P_0(\nu)
\\&\phantom{{}={}-}{}\cdot
\sum_{k=0}^\infty
\frac{(1+m)_k(-m)_k(1+l-|\nu|)_{k}(|\nu|-l)_{k}(z-4)^{2k+1}}{4(2k+1)!(2k)!}
,
\end{align*}
hence the third formula from Theorem~\ref{19080717} follows.
\end{proof}

\begin{remark*}
The transformation formulas for ${_4F_3}$ from \cite{B,BS} 
partially match our purpose on the expansion around $z=0,8$.
However the resulting identity is in fact much more complicated
than those asserted in Theorem~\ref{19080717}. 
\end{remark*}

\begin{proof}[Proof of Corollary~\ref{200103}]
A fundamental solution to $H_0$ can be obtain immediately by letting $z=0$ in
the analytic parts of $\mathcal P(m,l)$ and $\mathcal Q(m,l)$ from Theorem~\ref{19080717}. 
Their rational parts form another fundamental solution, 
and their irrational parts form a generalized eigenfunction, since $H_0$ does not mix them up. 
The irrational parts further split into two parts 
with factors $1/\pi$ and $(\log 2)/\pi$, each of which is again a generalized eigenfunction
by the same reason. 
If we remove the terms with factor $(\log 2)/\pi$, then we obtain 
the eigenfunction of the assertion. We are done.
\end{proof}

\appendix

\section{Comparison with previous results}\label{19121823}

\subsection{Result from the previous paper~\cite{IJ}}\label{190831}

In \cite{IJ} we computed the singular parts of the resolvent expansions
around all the thresholds in all the dimensions. 
They should coincide with the results of this paper.
For a quick comparison here we review the results from \cite{IJ}
according to the present setting only for the low dimensions.

\begin{theorem}\label{1908311827}
Let $d=1$. 
There exists $\mathcal E_0(z,n)$ analytic in $|z|<4$ such that 
for any $z\in\mathbb C\setminus [0,4]$ with $|z|<4$ and $n\in\mathbb Z$
$$G(z,n)
=
\mathcal E_0(z,n)+\frac1{2\sqrt{-z}}F^{(1)}_B\left(\genfrac{}{}{0pt}{}{\frac12+n;\frac12-n}{\frac12};\frac{z}4\right).
$$
Similarly, there exists $\mathcal E_1(z,n)$ analytic in $|z-4|<4$ such that 
for any $z\in\mathbb C\setminus [0,4]$ with $|z-4|<4$ and $n\in\mathbb Z$
$$G(z,n)
=
\mathcal E_1(z,n)+\frac{(-1)^{n+1}}{2\sqrt{z-4}}
F^{(1)}_B\left(\genfrac{}{}{0pt}{}{\frac12+n;\frac12-n}{\frac12};\frac{4-z}4\right).
$$
\end{theorem}

\begin{theorem}\label{1908311828}
Let $d=2$. 
There exists $\mathcal E_0(z,n)$ analytic in $|z|<4$ such that 
for any $z\in\mathbb C\setminus[0,8]$ with $|z|<4$ and $n\in\mathbb Z^2$
\begin{align*}
G(z,n)
&=
\mathcal E_0(z,n)-\frac1{4\pi}[\log (-z)]
F^{(2)}_B\left(\genfrac{}{}{0pt}{}{\frac12+n_1,\frac12+n_2;\frac12-n_1,\frac12-n_2}1;
\frac{z}4,\frac z4\right);
\intertext{There exists $\mathcal E_1(z,n)$ analytic in $|z-4|<4$ such that 
for any $|z-4|<4$ 
with $\mathop{\mathrm{Im}}z>0$ and $n\in\mathbb Z^2$}
G(z,n)
&=
\mathcal E_1(z,n)
-\frac{\mathrm i}{8\pi}\left[\log\left(-\frac{(z-4)^2}{16}\right)\right]
\\&\phantom{{}={}}{}\cdot
\biggl[
(-1)^{n_2}F^{(2)}_B
\left(\genfrac{}{}{0pt}{}{\frac12+n_1,\frac12+n_2;\frac12-n_1,\frac12-n_2}1;\frac{z-4}4,\frac{4-z}4\right)
\\&\phantom{{}={}\cdot\biggl[}{}
+
(-1)^{n_1}F^{(2)}_B\left(\genfrac{}{}{0pt}{}{\frac12+n_1,\frac12+n_2;\frac12-n_1,\frac12-n_2}1;\frac{4-z}4,\frac{z-4}4\right)
\biggr];
\intertext{There exists $\mathcal E_2(z,n)$ analytic in $|z-8|<4$ such that 
for any $z\in\mathbb C\setminus[0,8]$ with $|z-8|<4$ and $n\in\mathbb Z^2$}
G(z,n)
&=
\mathcal E_2(z,n)
+\frac{(-1)^{|n|}}{4\pi}[\log (z-8)]
\\&\phantom{{}={}}{}\cdot
F^{(2)}_B\left(\genfrac{}{}{0pt}{}{\frac12+n_1,\frac12+n_2;\frac12-n_1,\frac12-n_2}1;\frac{8-z}4,\frac{8-z}4\right).
\end{align*}
\end{theorem}

\subsection{Comparison of the singular parts}\label{190831b}

For $d=1$ the singular parts of Theorem~\ref{190716} and Theorem~\ref{1908311827} clearly coincide
due to \eqref{191221}.
Hence we examine only the case $d=2$. 
We begin with the embedded threshold $z=4$. 

\begin{corollary}\label{191020}
The singular parts from Theorems~\ref{19083122} and \ref{1908311828} around $z=4$ coincide,
which is equivalent to the following relations: 
For all $k\in\mathbb N_0$ and $n\in\mathbb Z^2$ 
\begin{align*}
&
\sum_{|\alpha|=2k}
\frac{(-1)^{\alpha_1}}{\alpha_1!\alpha_2!}
\left(\frac12+n_1\right)_{\alpha_1}
\left(\frac12+n_2\right)_{\alpha_2}
\left(\frac12-n_1\right)_{\alpha_1}
\left(\frac12-n_2\right)_{\alpha_2}
\\
&=
\frac{4^{2k}}{(2k)!}\left(\frac{1+n_1+n_2}2\right)_k
\left(\frac{1+n_1-n_2}2\right)_k
\left(\frac{1-n_1+n_2}2\right)_k
\left(\frac{1-n_1-n_2}2\right)_k
\intertext{and }
&\sum_{|\alpha|=2k+1}
\frac{(-1)^{\alpha_1}}{\alpha_1!\alpha_2!}\left(\frac12+n_1\right)_{\alpha_1}
\left(\frac12+n_2\right)_{\alpha_2}
\left(\frac12-n_1\right)_{\alpha_1}
\left(\frac12-n_2\right)_{\alpha_2}
\\
&=
\frac{4^{2k+1}}{(2k+1)!}\left(\frac{n_1+n_2}2\right)_{k+1}
\left(\frac{n_1-n_2}2\right)_{k+1}
\left(\frac{2-n_1+n_2}2\right)_k
\left(\frac{2-n_1-n_2}2\right)_k
.
\end{align*}
\end{corollary}

\begin{proof}
Since Theorems~\ref{19083122} and \ref{1908311828} provide 
asymptotic expansions of the same function, 
it is clear that their singular parts have to coincide. 

Let us verify that this is equivalent to the asserted identities. 
For that we set 
\begin{align*}
s_j[n]
&=
\frac18
\biggl[
(-1)^{n_2}
\sum_{|\alpha|=j}
\frac{
(-1)^{\alpha_2}\left(\frac12+n_1\right)_{\alpha_1}
\left(\frac12+n_2\right)_{\alpha_2}
\left(\frac12-n_1\right)_{\alpha_1}
\left(\frac12-n_2\right)_{\alpha_2}
}{j!\alpha_1!\alpha_2!}
\\&\phantom{{}={}\cdot\biggl[}{}
+
(-1)^{n_1}
\sum_{|\alpha|=j}
\frac{
(-1)^{\alpha_1}\left(\frac12+n_1\right)_{\alpha_1}
\left(\frac12+n_2\right)_{\alpha_2}
\left(\frac12-n_1\right)_{\alpha_1}
\left(\frac12-n_2\right)_{\alpha_2}
}{j!\alpha_1!\alpha_2!}
\biggr],
\end{align*}
which are the coefficients of the singular part from Theorem~\ref{1908311828}. 
First let $j=2k$ be even, and then  
\begin{align*}
s_{2k}[n]
&=
\frac{(-1)^{n_1}+(-1)^{n_2}}{8(2k)!}
\\&\phantom{{}={}}{}\cdot
\sum_{|\alpha|=2k}
\frac{
(-1)^{\alpha_1}
}{\alpha_1!\alpha_2!}\left(\frac12+n_1\right)_{\alpha_1}
\left(\frac12+n_2\right)_{\alpha_2}
\left(\frac12-n_1\right)_{\alpha_1}
\left(\frac12-n_2\right)_{\alpha_2}
.
\end{align*}
Then the coincidence of the singular parts implies 
the first identity of the assertion for $|n|$ even,
and hence also for $|n|$ odd. 
Next let $j=2k+1$ be odd, and we have  
\begin{align*}
s_{2k+1}[n]
&=
\frac{(-1)^{n_1}-(-1)^{n_2}}{8(2k+1)!}
\\&\phantom{{}={}}{}\cdot
\sum_{|\alpha|=2k+1}
\frac{(-1)^{\alpha_1}}{\alpha_1!\alpha_2!}
\left(\frac12+n_1\right)_{\alpha_1}
\left(\frac12+n_2\right)_{\alpha_2}
\left(\frac12-n_1\right)_{\alpha_1}
\left(\frac12-n_2\right)_{\alpha_2}
.
\end{align*}
Similarly to the above, the coincidence of the singular parts implies 
the first asserted identity for $|n|$ odd,
and hence also for $|n|$ even. 
The converse is clear from the above argument, and we are done.
\end{proof}
\begin{remark*}
We do not try to give direct proofs of the identities from Corollary~\ref{191020}.
The identities in Corollary~\ref{191020} have been verified using the computer algebra system Maple.
\end{remark*}

Similarly we discuss the endpoint thresholds $z=0,8$ as follows. 

\begin{corollary}\label{191021}
The singular parts from Theorems~\ref{19080717} and \ref{1908311828} around 
$z=0,8$ coincide, i.e., for all $|w|<1$ and $(m,l)\in\mathbb Z^2$ 
\begin{align*}
&(-1)^{m+l}
F^{(2)}_B\left(\genfrac{}{}{0pt}{}{\frac12+m+l,\frac12+m-l;\frac12-m-l,\frac12-m+l}1;w,w\right)
\\&
=
{_2F_1}\left(\genfrac{}{}{0pt}{}{\frac12,\frac12}{1};w(2-w)\right)
{_4F_3}\left(\genfrac{}{}{0pt}{}{m,-m,l,-l}{1,\frac12,\frac12};(w-1)^2\right)
\\&\phantom{{}={}}{}
+\sum_{\mu=1}^{|m|}
(-1)^\mu
\left[{_2F_1}\left(\genfrac{}{}{0pt}{}{\frac12+\mu,\frac12-\mu}{1};w(2-w)\right)
+{_2F_1}\left(\genfrac{}{}{0pt}{}{\mu-\frac12,\frac32-\mu}{1};w(2-w)\right)\right]
\\&\phantom{{}={}}\qquad{}\cdot
{_4F_3}\left(\genfrac{}{}{0pt}{}{1+|m|-\mu,\mu-|m|,l,-l}{1,\frac12,\frac12};(w-1)^2\right)
\\&\phantom{{}={}}{}
+\sum_{\nu=1}^{|l|}
(-1)^\nu
\left[
{_2F_1}\left(\genfrac{}{}{0pt}{}{\frac12+\nu,\frac12-\nu}{1};w(2-w)\right)
+{_2F_1}\left(\genfrac{}{}{0pt}{}{\nu-\frac12,\frac32-\nu}{1};w(2-w)\right)\right]
\\&\phantom{{}={}}\qquad{}\cdot
{_4F_3}\left(\genfrac{}{}{0pt}{}{m,-m,1+|l|-\nu,\nu-|l|}{1,\frac12,\frac12};(w-1)^2\right)
,
\intertext{and for all $|w|<1$ and $(m,l)\in\mathbb N_0^2$}
&(-1)^{m+l}
F^{(2)}_B\left(\genfrac{}{}{0pt}{}{\frac32+m+l,\frac12+m-l;-\frac12-m-l,\frac12-m+l}1;w,w\right)
\\&
={}
(2m+1)(2l+1)(w-1)
{_2F_1}\left(\genfrac{}{}{0pt}{}{\frac12,\frac12}{1};w(2-w)\right)
\\&\phantom{{}={}}\qquad{}\cdot
{_4F_3}\left(\genfrac{}{}{0pt}{}{1+m,-m,1+l,-l}{1,\frac32,\frac32}; (w-1)^2\right)
\\&\phantom{{}={}}{}
-
(2l+1)(w-1)
\sum_{\mu=-m}^{m}(-1)^\mu
{_2F_1}\left(\genfrac{}{}{0pt}{}{\frac12+\mu,\frac12-\mu}{1};w(2-w)\right)
\\&\phantom{{}={}}\qquad{}\cdot
{_4F_3}\left(\genfrac{}{}{0pt}{}{1+m-|\mu|,|\mu|-m,1+l,-l}{1,\frac12,\frac32}; (w-1)^2\right)
\\&\phantom{{}={}}{}
-
(2m+1)(w-1)
\sum_{\nu=-l}^{l}(-1)^\nu
{_2F_1}\left(\genfrac{}{}{0pt}{}{\frac12+\nu,\frac12-\nu}{1};w(2-w)\right)
\\&\phantom{{}={}}\qquad{}\cdot
{_4F_3}\left(\genfrac{}{}{0pt}{}{1+m,-m,1+l-|\nu|,|\nu|-l}{1,\frac12,\frac32}; (w-1)^2\right)
.
\end{align*}
\end{corollary}

\begin{proof}
We can prove the assertion similarly to Corollary~\ref{191020}. 
However, it requires fairly long computations, and we omit it. 
\end{proof}
\begin{remark*}
In contrast to Corollary~\ref{191020} we have no idea how 
we could possibly directly prove the identities from Corollary~\ref{191021}. 
However, they can be verified using the computer algebra system Maple, in this case for fixed but arbitrary values of $m$ and $l$.
\end{remark*}

\section{Renormalized expectation of random work}\label{191219}

\begin{proof}[Proof of \eqref{190805}]
It is elementary to compute $P(X_k=n)$. In fact, we have 
\begin{align*}
P(X_k=n)
=\frac1{(2d)^k}\sum_{|\alpha|=(k-|n|)/2}\frac{k!}{\alpha!\prod_{j=1}^d\bigl[(\alpha_j+|n_j|)!\bigr]}
\end{align*}
if $k\ge |n|$ and $k-|n|$ is even, and 
\begin{align*}
P(X_k=n)=0
\end{align*}
otherwise.
Then we obtain 
\begin{align*}
\sum_{k=0}^\infty \left(\frac{2d}{2d-z}\right)^kP(X_k=n)
&
=
\sum_{\genfrac{}{}{0pt}{}{k\ge |n|,}{k-|n|\text{:even}}} 
\sum_{|\alpha|=(k-|n|)/2}
\frac{k!}{\alpha!\prod_{j=1}^d\bigl[(\alpha_j+|n_j|)!\bigr]}
(2d-z)^{-k}
\\&
=
\sum_{\alpha\in\mathbb N_0^d} 
\frac{(2|\alpha|+|n|)!}{\alpha!\prod_{j=1}^d\bigl[(\alpha_j+|n_j|)!\bigr]}
(2d-z)^{-2|\alpha|-|n|},
\end{align*}
and hence \eqref{190805} follows by Theorem~\ref{181110}.
\end{proof}

\bigskip
\noindent
\subsubsection*{Acknowledgements} 
KI would like to thank Professors Naotaka Kajino, 
Hideshi Yamane and Kanam Park for valuable comments. 
This work was initiated and finalized during KI's visits to Aarhus University,
and he would like to thank for their kind hospitality. 
KI was partially supported by JSPS KAKENHI Grant Number 17K05325.
The authors were partially supported by the Danish Council for Independent Research $|$ Natural Sciences, Grants DFF--4181-00042 and DFF--8021-0084B.

\end{document}